\documentclass[11pt]{article}

\usepackage{amsmath,amsthm,amssymb,amsfonts,verbatim}
\usepackage[colorlinks,
            linkcolor=blue,
            anchorcolor=blue,
            citecolor=blue
            ]{hyperref}
\usepackage[hmargin=1.2in,vmargin=1.2in]{geometry}

\date{}

\newtheorem{lemma}{Lemma}
\newtheorem{theorem}{Theorem}

\newtheorem{ex}{Example}

\newtheorem{defn}{Definition}

\renewcommand{\epsilon}{\varepsilon}
\renewcommand{\le}{\leqslant}
\renewcommand{\ge}{\geqslant}
\def\lpf {{\rm lpf}}

\newcommand{\vnote}[1]{}

\def\F{\mathbb{F}}
\def \mC {\mathcal{C}}

\def \mC {\mathcal{C}}

\def \mF {\mathcal{F}}

\def \mS {\mathcal{S}}

\def \mX {\mathcal{X}}

\def \Xi {{X^{[i]}}}

\newcommand{\Ga}{\alpha}

\def \bx {{\bf x}}

\def \by {{\bf y}}
\def \bs {{\bf s}}

\def\m {{\rm m }}

\begin{document}

\title{A Construction of Optimal Frequency Hopping Sequence Set via Combination of Multiplicative and Additive Groups of Finite Fields}

\author{Xianhua Niu \thanks{School of Computer and Software Engineering, Xihua University and National Key Laboratory of Science and Technology on Communications, University of Electronic Science and Technology of China. Research supported by the Youth Science and Technology
Fund of Sichuan Province (No.2017JQ0059). Email: {\tt rurustef1212@gmail.com.}} \and Chaoping Xing\thanks{School of Physical and Mathematical Sciences, Nanyang Technological University, Singapore. Email: {\tt xingcp@ntu.edu.sg}.} }

\maketitle

\begin{abstract}
In literatures, there are various constructions of frequency hopping sequence (FHS for short) sets with good Hamming correlations. Some papers employed only multiplicative groups of finite fields to construct FHS sets, while  other papers implicitly used only additive groups of finite fields for construction of FHS sets.
In this paper, we make use of both multiplicative  and additive groups of finite fields simultaneously to present a construction of optimal FHS sets. The construction provides a new family of optimal $\left(q^m-1,\frac{q^{m-t}-1}{r},rq^t;\frac{q^{m-t}-1}{r}+1\right)$ frequency hopping sequence sets archiving the Peng-Fan bound. Thus, the  FHS sets constructed in literatures using either multiplicative groups or additive groups of finite fields are all included in our family.
In addition, some other FHS sets can be obtained via the well-known recursive constructions through one-coincidence sequence set.
\end{abstract}

\section{Introduction}
In frequency-hopping multiple access (FHMA) communication systems, each user's wideband signal is generated by hopping over a large number of frequency slots.
User's frequency slots used are chosen pseudo-randomly via a code called frequency hopping sequences. The degree of the mutual interference between users is clearly related to the Hamming correlation properties of the frequency hopping sequences, and the number of users allowed by the system for synchronous communication is determined by the number of frequency hopping sequences \cite{Fan,Simon}. In order to improve the system performance, it is desirable to employ frequency hopping sequences having low Hamming correlation to reduce the multiple-access interference (also called hits) of frequencies \cite{LG74}. Moreover, the required sequence length and alphabet size of FHS set are variable according to the specification of a given system or environment. Thus, the design of FHS set with good Hamming correlation property and flexible parameters is an important problem.

In practical applications, the required length and alphabet size of an FHS or an FHS set vary depending on the specification of a given system or environment. Thus, it is very important to select  FHS sets with optimal Hamming correlation under the given condition. In general, optimality of an FHS set is measured by the Peng-Fan bound ~\cite{PF}, whereas that of a single FHS is by the Lempel-Greenberger bound~\cite{LG74}. It is of particular interest to construct FHS sets which meet Peng-Fan bound. There are several algebraic method, combinatorial and recursive constructions in the literature~\cite{LG74,U98,Chu05,Ding08,Zhou11,Zeng12,Zeng13,Chung14,Ren14,Xu16,Bao16,Niu18,Xu18}.
\subsection{Known results}\label{subsec:1.1}
In literatures, there are various constructions of FHS sets with good Hamming correlations. Let us only recall the constructions relevant to our construction, namely those via either multiplicative  or additive groups of finite fields.
\begin{itemize}
\item[(1)] Lempel et al.~\cite{LG74,U98,Zhou11} showed that there is  an optimal $(q^m-1,q^{m-t},q^t;q^{m-t})$ FHS set for a prime power $q$ and integers $1\leq t\leq m-1$. This construction implicitly employed the additive group structure of a finite field.

\item[(2)]
Ding et al.~\cite{Ding08} constructed optimal $\left(q-1,\frac{q-1}{f},f;\frac{q-1}{f}+1\right)$ FHS set for a prime power $q$ and integer $f$ satisfying $f|(q-1)$, and $2\leq f\leq \frac{q-1}{f}-1$. This construction used  the multiplicative group structure of a finite field.
\end{itemize}

\subsection{Our result}

The optimal FHS sets given in Subsection \ref{subsec:1.1} were constructed via either multiplicative or additive group structure of finite fields. By mixing both multiplicative and additive group structures of finite fields, we obtain the following result.

\begin{theorem}\label{thm:1.1} Let $q$ be a prime power and let $r$ be a divisor of $q-1$. Then for any  $0\leq t\leq m-1$,  there is an  FHS set $\mS$ with parameters
\[ \left\{\begin{array}{ll}\left(q^m-1,\frac{q^{m-t}-1}{r},rq^t;\frac{q^{m-t}-1}{r}+1\right)& \mbox{if $r\ge 2$}\\
\left(q^m-1,1+\frac{q^{m-t}-1}{r},rq^t;\frac{q^{m-t}-1}{r}+1\right)=(q^m-1,q^{m-t},q^t;q^{m-t})& \mbox{if $r=1$} \end{array} \right.\]
In addition, if $q^m-1<e^2+(e+1)q^t-3e$ with $e=\frac{q^{m-t}-1}r$, then $\mS$ is optimal, i.e., it achieves the Peng-Fan bound.
\end{theorem}
It is easy to see that by taking $t=0$, we get the FHS set given in Ding et al. ~\cite{Ding08}. By taking $r=1$, we get the FHS set given in Lempel et al.\cite{LG74,U98,Zhou11}. Compared with the constructions in \cite{Ding08,LG74,U98,Zhou11},
the FHS set in Theorem \ref{thm:1.1} allows more flexible parameters due to the free choice of $t$.

By applying the FHS sets in Theorem \ref{thm:1.1} to the standard recursive concatenation of   FHS sets with  one-coincidence (OC for short) sequence sets, we obtain the following new FHS sets.
\begin{theorem}\label{thm:1.2} Let $q$ be a prime power and let $r\geq2$ be a divisor of $q-1$. Then for any  $0\leq t\leq m-1$, we have FHS sets given in the following table.
\begin{table}[t]\label{tab1}\footnotesize
\newcommand{\tabincell}[2]{\begin{tabular}{@{}#1@{}}#2\end{tabular}}
\centering  % ±í¾ÓÖÐ
\caption{Parameters of some new recursive constructions of optimal FHS Sets }
\medskip

\begin{tabular}{c|c|c|c|l}  %p{3cm} {lccc} ±íʾ¸÷ÁÐÔªËضÔÆ뷽ʽ£¬left-l,right-r,center-c
\hline
Length $N$ & $M$ & $H_m$ & Alphabet Size $\ell$ & Constraints \\ \hline\hline  % \hline ÔÚ´ËÐÐÏÂÃæ»­Ò»ºáÏß
$k(q^m-1)$ & $e$ & $rq^t$ & $k(e+1)$ &\tabincell{l}{$er+1=q^{m-t}, 1\leq t\leq m-1$,\\$q^m-q^t-1<\lpf(k)$, \\ $q^m-1<e^2+(e+1)p^t-3e$}
%& \tabincell{l}{$(n,\lpf(n)-1;n)$ \\~\cite{Bao16,Lenny17,Ren17,Niu18}}
\\ \hline
$(p-1)(q^m-1)$ & $e$ & $rq^t$ & $p(e+1)$ &\tabincell{l}{$er+1=q^{m-t}, 1\leq t\leq m-1$,\\$q^m-q^t-1\leq p$ ,\\ $q^m-1<e^2+(e+1)p^t-3e$}
%& \tabincell{l}{$(p-1,p;p)$ \\~\cite{OC,Chung14,Bao16,Niu18}}
\\ \hline
$k(p-1)(q^m-1)$ & $e$ & $rq^t$ & $kp(e+1)$ &\tabincell{l}{$er+1=q^{m-t}, 1\leq t\leq m-1$,\\$q^m-q^t-1\leq \min\{\lpf(k)-1,p\}$,\\ $q^m-1<e^2+(e+1)p^t-3e$ }
%& \tabincell{l}{$(n(p-1),M;nq_1)$ \\ $M=\min\{\lpf(n)-1,p\}$,\\~\cite{Chung14,Niu18}}
\\ \hline \hline
\end{tabular}
\medskip
\\ In  the above table, $p$ denotes a prime power and $\lpf(k)$ denotes the least prime factor of an integer $k>1$.
\end{table}
\end{theorem}
\subsection{Our techniques}
Our approach is divided into four steps:
\begin{itemize}
\item[(i)] partition a finite field $\F_{q^m}$ into $\ell$ disjoint subsets $V_i$ for $i=1,2,\dots,\ell$;
\item[(ii)] construct a polynomial  $\phi(x)$ that is a constant polynomial in each $V_i$;
\item[(iii)] for every $b\in\F_{q^m}$, define an FHS $\bs_b:=\left(\phi(\theta^0+b),\phi(\theta+b),\phi(\theta^2+b),\dots,\phi(\theta^{q^m-2}+b)\right)$, where $\theta$ is a primitive element of $\F_{q^m}$;
\item[(iv)] choose a suitable subset $S$ of $\F_{q^m}$  to  form  an  FHS set $\{\bs_b:\; b\in S\}$.
\end{itemize}
 The key part of the above approach is partition of $\F_{q^m}$ into $\ell$ disjoint subsets $V_i$ for $i=1,2,\dots,\ell$.  Although different languages were adopted to construct FHS sets in \cite{Ding08,LG74,U98,Zhou11}, they implicitly used the above approach. More precisely speaking,  Ding et al. ~\cite{Ding08}  partitioned $\F_{q^m}$ into disjoint cosets of a multiplicative group of $\F_{q^m}$, while  Lempel et al.\cite{LG74,U98,Zhou11}  partitioned $\F_{q^m}$ into disjoint cosets of an $\F_q$-subspace   of $\F_{q^m}$.

 In this paper, we choose a multiplicative group $G$ of $\F_{q^m}$ and  an $\F_q$-subspace $V$  of $\F_{q^m}$, then  partition $\F_{q^m}$ into disjoint cosets of  by mixing the structures of $G$ and $V$. Thus, for the trivial group $G=\{1\}$, it degenerates to an $\F_q$-subspace $V$  of $\F_{q^m}$, while for the trivial vector space $V=\{0\}$ and $m=1$, it degenerates to a multiplication subgroup of $\F_{q^m}$.

\subsection{Organization}

The rest of this paper is organized as follows. In Section 2, we give some preliminaries to frequency hopping sequences. In Section 3, we present a construction of optimal FHS sets with new parameters by mixing both multiplicative and additive group structures of finite fields. In Section 4, we obtain some optimal FHS sets via the recursive construction through one-coincidence sequence sets. Finally, we conclude the paper in Section 5.
\section{Preliminaries}
Let $\mF$= \{$f_{1}$, $f_{2}$,\ldots, $f_{\ell}$\} be a frequency slot set with size $|\mF|$$=$$\ell$, and let $\mS$ be a set of $M$ frequency hopping sequences of length $N$. For any two frequency hopping sequences $\bx$=($x_{0}$, $x_{1}$,\ldots, $x_{N-1}$) and $\by$=($y_{0}$, $y_{1}$,\ldots, $y_{N-1}$)$\in\!\mS$, and any positive integer $\tau$, $0 \!\leq\!\tau\!\leq\!N-1$, the Hamming correlation function $H_{\bx\by}(\tau)$ of $\bx$ and $\by$ at time delay
$\tau$ is defined as follows:
\begin{eqnarray}
H_{\bx\by}(\tau)=\sum^{N-1}_{i=0}h(x_{i},y_{i+\tau}),
\end{eqnarray}
where $h(a, b)=1$ if $a=b$, and $h(a, b)=0$ otherwise. %And only positive time shifts are considered.

For a given FHS set $\mS$, the maximum Hamming autocorrelation $H_{a}(\mS)$, the maximum Hamming crosscorrelation $H_{c}(\mS)$ and the maximum Hamming
correlation $H_{m}(\mS)$ are defined as follows, respectively:
\begin{eqnarray*}
H_{a}(\mS)\!&\!=\!&\max\limits_{1\!\leq\!\tau\!\leq\!N-1}\{H_{\bx\bx}(\tau): \bx\!\in\!\mS\},\\
H_{c}(\mS)\!&\!=\!&\max\limits_{0\!\leq\!\tau\!\leq\!N-1}\{H_{\bx\by}(\tau): \bx,\by\in\!\mS,\bx\!\neq \!\by\},\\
H_{m}(\mS)\!&\!=\!&\max\{H_a(\mS),H_c(\mS)\}.
\end{eqnarray*}
%We denote by $(N,M,\lambda,\ell)$ an FHS set over a frequency slot set with size $\ell$ of length $N$, size $M$ and Hamming correlation $\lambda$.
In 2004, Peng and Fan~\cite{PF} showed that  the maximum Hamming correlation $H_m(\mS)$ of an    FHS set $\mS$ of $M$ sequences of length $N$ over a frequency slot set of size $\ell$ must obey
\begin{equation}\label{e2}
 H_m(\mS)\geq \left\lceil\frac{(NM-\ell)N}{(NM-1)\ell}\right\rceil.
\end{equation}
In this paper, an FHS set $\mS$ is said {\it optimal} if it achieves the Peng-Fan bound with equality.

The one-coincidence sequence set is a special FHS set which was proposed firstly by Shaar and Davies \cite{OC} in 1984.

\begin{defn}\label{OC}{\rm
A one-coincidence sequence set is a set of nonrepeating sequences, for which the peak of the Hamming crosscorrelation function equals one for any pair of sequences belonging to the set.
}\end{defn}
Equivalently speaking,  an OC sequence set is an FHS set with the maximum Hamming autocorrelation equal to $0$  and maximum Hamming crosscorrelation at most $1$.				

The following notations will be used throughout this paper:
\begin{itemize}
\item $ (N,M,\lambda;\ell)$ denotes an FHS set of $M$ sequences of length $N$ over a frequency slot set of size $\ell$, with the maximum Hamming correlation equals to $\lambda$;
\item
$ (n,s;v)$ denotes an OC sequence set of $s$ sequences of length $n$ over a frequency slot set of size $v$, with the maximum Hamming autocorrelation equal to $0$ and the maximum Hamming corsscorrelation at most $1$;
%\item
%$  p$ is a prime power and $q$ is a prime power of $p$ with $p<q$;
\item $r$ is a divisor of $q-1$;
\item $0\le t\le m-1$ are integers;
%\item $\lpf(n)$ is the least prime factor of positive integer $n$;
\item
$\lceil z\rceil$ is the smallest integer larger than or equal to a real number $z$.
\end{itemize}

\section{New construction of FHS Sets }
Let $G$ be a multiplicative subgroup $\F_q^*$ with $|G|=r$. We label all elements of $G=\{g_1,g_2,\cdots, g_r\}$.
Let $V$ be an $\F_{q}$-subspace of dimension $t$. Then $|V|=q^t$.

\begin{lemma}\label{l1}
There exist $ \alpha_1=0$, $\alpha_2,\cdots, \alpha_\ell\in \F_{q^m}$ with $\ell=1+\frac{q^{m-t}-1}{r}$ such that $V$ and $\{\alpha_ig+V:{2\leq i\leq \ell, g\in G}\}$ are $1+(\ell-1)r=q^{m-t}$ pairwise distinct cosets of $V$.
\end{lemma}
\begin{proof}
Choose $\alpha_2\in \F_{q^m}\setminus V$. We claim  that $\alpha_2g_1+V, \cdots, \alpha_2g_r+V$ are pairwise distinct cosets of $V$. Suppose $\alpha_2g_i+V=\alpha_2g_j+V$ for some $1\leq i\le j\leq r$. Then, we have $\alpha_2(g_i-g_j)\in V$. Thus, $g_i-g_j=0$, otherwise, one would have $\alpha_2\in(g_i-g_j)^{-1}V=V$. This implies that $i=j$.

Next we choose $\alpha_3\in \F_{q^m}\setminus V\bigcup(\bigcup \limits_{g\in G}(\alpha_2g+V))$.
Then, in the same way, we can show that $\alpha_3g_1+V, \cdots, \alpha_3g_r+V$ are pairwise distinct cosets.
Furthermore, we  claim that $\alpha_2g_i+V\neq\alpha_3g_j+V$ for all $1\leq i\leq j\leq r$. Suppose $\alpha_2g_i+V=\alpha_3g_j+V$ for a pair $(i,j)$ with $1\leq i\leq j\leq r$, then we would  have $\alpha_2g_i-\alpha_3g_j\in V$, i.e., $\alpha_2g_ig_j^{-1}-\alpha_3\in g_j^{-1}V=V$. Thus, we have $\alpha_3\in \alpha_2g_ig_j^{-1}+V\subseteq \bigcup\limits_{g\in G}\alpha_2g+V$. This is a contradiction.

Continue this fashion to choose  $\alpha_i\in \F_{q^m}\setminus V\bigcup(\bigcup\limits_{2\le j\le i-1}\bigcup\limits_{g\in G}(\alpha_jg+V))$ for $i\ge 4$. Thus, we obtain all the desired cosets $V$ and $\{\alpha_ig+V:{2\leq i\leq \ell, g\in G}\}$.\end{proof}

\begin{lemma}\label{l2}
Define $\phi(x)=\prod\limits_{g\in G}\prod\limits_{\beta\in V}(x+g+\beta)$. Then $\phi(x)$ is a constant function on the set $\bigcup\limits_{i=1}^{r}(\gamma g_i+V)$ for any fixed $\gamma\in \F_{q^m}$.
\end{lemma}
\begin{proof}
Let $\gamma \Ga+v\in \bigcup\limits_{i=1}^{r}(\gamma g_i+V)$ for some $\Ga\in G$ and $v\in V$.

Then we have
\begin{eqnarray*}
\phi(\gamma \Ga+v)&=&\prod\limits_{g\in G}\prod\limits_{\beta\in V}(\gamma \Ga+v+g+\beta)=\Ga^{|G||V|}\prod\limits_{g\in G}\prod\limits_{\beta\in V}(\gamma +g\Ga^{-1}+\beta \Ga^{-1})\\
&=&\Ga^{rq^t}\prod\limits_{g\in G}\prod\limits_{\beta\in V}(\gamma +g+\beta \Ga^{-1})=\prod\limits_{g\in G}\prod\limits_{\beta\in V}(\gamma +g+\beta)=\phi(\gamma).
\end{eqnarray*}

Thus, $\phi(x)$ is a constant function on the set $\bigcup\limits_{i=1}^{r}(\gamma g_i+V)$ for any fixed $\gamma\in \F_{q^m}$. This completes the proof.
\end{proof}

\begin{lemma}\label{l3}
If $1\leq i\neq j\leq \ell$ with $\ell=1+\frac{q^{m-t}-1}{r}$, then $\phi(\alpha)\neq\phi(\beta)$ for all $\alpha\in \bigcup\limits_{g\in G}(\alpha_i g+V)$, $\beta \in \bigcup\limits_{g\in G}(\alpha_j g+V)$.
\end{lemma}
\begin{proof}
Suppose $\phi(\alpha)=\phi(\beta)$. Let $c=\phi(\alpha)$. Consider the polynomial $\phi(x)-c$. Then all elements of $\Big(\bigcup\limits_{g\in G}(\alpha_i g+V)\Big)\bigcup \Big( \bigcup\limits_{g\in G}(\alpha_j g+V)\Big)$ are roots of $\phi(x)-c$. Thus, $\phi(x)-c$ has at least $\left|\Big(\bigcup\limits_{g\in G}(\alpha_i g+V)\Big)\bigcup \Big( \bigcup\limits_{g\in G}(\alpha_j g+V)\Big)\right|\ge |V|+|G||V|=q^t+rq^t$ roots. On the other hand, the degree of $\phi(x)-c$ is $rq^t$. This forces that $\phi(x)-c$ is identical $0$, i.e, $\phi(x)=c$. This contradiction completes the proof.
\end{proof}

%\begin{lemma}\label{l4}
%We have $\bigcap_{i=1}^\ell(\Ga_i\F_p+V)=V$.
%\end{lemma}
%\begin{proof} %If $t=0$, i.e., $V=\{0\}$, then $\bigcap_{i=1}^\ell(\Ga_i\F_q+V)=$\bigcap_{i=1}^\ell \Ga_i\F_q$
%It is clear that there exists $2\le i\le \ell$ such that $\Ga_i\not\in \F_p$. Otherwise, one would have $\F_{q^m}=\bigcup_{i=1}^\ell(\Ga_i\F_p+V)=\F_p+V$. This give $q^m=|\F_{q^m}|=|\F_p+V|\le|\F_p|\times|V|=pq^t\le pq^{m-1}<q^m$. This is a contradiction.
%\end{proof}

%\begin{cor}\label{c1}
%If $\phi(\alpha)=\phi(\beta)$ for $\alpha,\beta \in \F_{q^m}$,  and $\alpha\in \bigcup\limits_{g\in G}(\alpha_i g+V)$, then  $\beta \in \bigcup\limits_{g\in G}(\alpha_i g+V)$ for $1\leq i\leq \ell$.
%\end{cor}
%This corollary can be obtain from Lemma \ref{l3} easily, we omit the proof here.

%\begin{cons}\label{cons1}
\textbf{Construction of optimal FHS sets.}
\begin{enumerate}
\itemindent12pt
\item[Step {\rm1}:] Let $r$ be a divisor of $q-1$. Let $\theta$ be a generator of $\F_{q^m}^*$. Let $G=\F_{q}^*$ be a multiplicative subgroup of $\F_q$ with $|G|=r$. Let $V$ be an $\F_{q}$-subspace of $\F_{q^m}$ of dimension $t$ with $ 0\leq t\leq m-1$. Choose $\alpha_1=0$, $\alpha_2,\cdots, \alpha_\ell\in \F_{q^m}$ with $\ell=1+\frac{q^{m-t}-1}{r}$ satisfying that $V$, $\{\alpha_ig+V:{2\leq i\leq \ell, g\in G}\}$ are $q^{m-t}$ pairwise distinct cosets of $V$ as defined in Lemma \ref{l1}.

\item[Step {\rm2}:] Choose $\phi(x)$ as defined in Lemma \ref{l2}. For every $ \Ga_i\in\{\Ga_1,\Ga_2,\cdots, \Ga_\ell\}$, we  define a sequence
\[\bs_i:=(\phi(1+\Ga_i),\phi(\theta+\Ga_i),\phi(\theta^2+\Ga_i),\cdots, \phi(\theta^{q^m-2}+\Ga_i)).\]

\item[Step {\rm3}:] The desired FHS set $\mS$ is a collection of $\bs_i$, i.e.,
\[\mS=\left\{\begin{array}{ll}
\{\bs_i\}_{i=1}^\ell&\mbox{if $r=1$,}\\
\{\bs_i\}_{i=2}^\ell&\mbox{if $r\ge 2$.}
\end{array}
\right.\]
\end{enumerate}
%\end{cons}

\begin{theorem}\label{FHS1}
The FHS set $\mS$ constructed above  is a $\left(q^m-1,M,rq^t;\frac{q^{m-t}-1}{r}+1\right)$ FHS set with
\[M=\left\{\begin{array}{ll}
\ell=1+\frac{q^{m-t}-1}{r}=q^{m-t}&\mbox{if $r=1$,}\\
\ell-1=\frac{q^{m-t}-1}{r}&\mbox{if $r\ge 2$.}
\end{array}
\right.\]
\end{theorem}
\begin{proof} The sequence length of $\mS$ is clearly $q^m-1$. The family size of $\mS$ is also clear.
Furthermore, it follows from  Lemmas \ref{l2} and \ref{l3} that the size of frequency slot set of $\mS$ is $\ell=1+\frac{q^{m-t}-1}{r}$.

Thus, it is sufficient to show that the  maximum Hamming correlation of $\mS$ is $rq^t$. Given the facts: (i) the Hamming correlation $H_{\bs_i\bs_j}(\tau)$  at time delay
$\tau$ is the number of the roots of $\phi(\theta^\tau x+\Ga_i)- \phi(x+\Ga_j)$; and (ii) the degree of $\phi(\theta^\tau x+\Ga_i)- \phi(x+\Ga_j)$ is at most $rq^t$, it is equivalent to showing  the following two inequalities.
\begin{itemize}
\item[(a)] $\phi(\theta^\tau x+\Ga_i)\neq \phi(x+\Ga_i), \quad \mbox{$\forall 1\leq \tau\leq q^m-2$ and $2\le i\le \ell$ if $r\ge 2$ (and $1\le i\le \ell$ if $r=1$)}$.
\item[(b)] $\phi(\theta^\tau x+\Ga_i)\neq \phi(x+\Ga_j), \quad \mbox{$\forall 0\leq \tau\leq q^m-2$ and $1\le i<j\le \ell$}.$
\end{itemize}
Let us prove (a) by contradiction. Suppose $\phi(\theta^\tau x+\Ga_i)= \phi(x+\Ga_i)$ for some $1\leq \tau\leq q^m-2$. Then, by comparing the leading coefficients of $\phi(\theta^\tau x+\Ga_i)$ and $ \phi(x+\Ga_i)$, we have $(\theta^\tau)^{r}=1$, i.e., $\theta^\tau \in G$. If $r=1$, i.e, $G=\{1\}$, then  $\theta^\tau=1$. This is a contradiction since $\theta^\tau\neq1$ for all $1\leq \tau\leq q^m-2$.

Now we assume that $r\ge 2$.
Then $2\le i\le \ell$ and  we have
\begin{eqnarray*}
 \phi(x+\Ga_i)=\phi(\theta^\tau x+\Ga_i)=(\theta^\tau)^{r}\phi(x+\Ga_i\theta^{-\tau})=\phi(x+\Ga_i\theta^{-\tau}).
\end{eqnarray*}

Choose $\gamma\in \F_{q^m}$ such that $\gamma+\Ga_i \in V$. By Lemma \ref{l3}, we must have $\gamma+\Ga_i\theta^{-\tau} \in V$.
This gives $(\gamma+\Ga_i\theta^{-\tau})-(\gamma+\Ga_i)\in V$, i.e. $\Ga_i(\theta^{-\tau}-1) \in V$. As $\theta^{-\tau}\neq 1$, i.e., $\theta^{-\tau}-1\in\F_q^*$, we have $\Ga_i \in (\theta^{-\tau}-1)^{-1}V=V$. This is a contradiction by Lemma \ref{l1}.

Again, we prove (b) by contradiction. Suppose $\phi(\theta^\tau x+\Ga_i)= \phi(x+\Ga_j)$ for some $0\leq \tau\leq q^m-2$ and $1\le i<j\le \ell$. By comparing the leading coefficients with the same arguments given in the above proof of (a), we have $\theta^\tau\in G$ and
$
\phi(x+\Ga_i\theta^{-\tau})= \phi(x+\Ga_j).
$

Choose $\gamma\in \F_{q^m}$ such that $\gamma+\Ga_j\in V$. By Lemma \ref{l3}, we have $\gamma+\Ga_i\theta^{-\tau} \in V$.
This gives $(\gamma+\Ga_i\theta^{-\tau})-(\gamma+\Ga_j)\in V$, i.e., $\Ga_j\in \Ga_i\theta^{-\tau}+V\subseteq\bigcup_{g\in G}(\Ga_ig+V)$. This is a contradiction by Lemma \ref{l1}. The proof is completed.
\end{proof}

{\bf Proof of Theorem \ref{thm:1.1}:}
\begin{proof} The first part of Theorem \ref{thm:1.1} was already proved in Theorem \ref{FHS1}. Let us prove the second part only.

By the Peng-Fan bound, the FHS set $\mS$ is optimal if and only if the following inequality is satisfied.
\begin{equation}\label{eq:1}
\frac{e(q^m-1)-(e+1)}{e(q^m-1)-1}\cdot\frac{q^m-1}{e+1}>rq^t-1.
\end{equation}
The inequality \eqref{eq:1} is equivalent to the following inequality.
\begin{equation}\label{eq:2}
e(q^m-1)^2-(e^3+(e^2+e)(q^t-1)+1)(q^m-1)+(e+1)(q^t-1+e)>0.
\end{equation}
The inequality \eqref{eq:2}  is always true if $q^m-1<e^2+(e+1)q^t-3e$.  This completes the proof.
\end{proof}

%The Construction \ref{cons1} combine  multiplicative group and additive group to obtain optimal FHS sets. If $t=0$ in Construction \ref{cons1}, we can obtain the FHS set based on multiplicative group only.

We now illustrate our construction  by the following examples.
\begin{ex}\label{e1}{\rm
Let $q=3$, $r=2$, $m=4$, and $t=1$.

Let $\theta$ be a generator of $\F_{3^4}^*$. Let $G=\F_{3}^*$ be the multiplicative subgroup of $\F_3$ with $|G|=2$. Let $V$ be an $\F_{3}$-subspace of $\F_{3^4}$ of dimension $t=1$.

Choose $\alpha=$$\{\alpha_i\in \F_{3^4}:2\leq i\leq14\}$ satisfying that $V$, $\{\alpha_ig+V:{2\leq i\leq 14, g\in G}\}$ are $27$ pairwise distinct cosets of $V$ as defined in Lemma \ref{l1}.

Choose $\phi(x)$ as defined in Lemma \ref{l2}. For every $ \Ga_i\in\{\Ga_2,\cdots, \Ga_{14}\}$, we can obtain an FHS set $\mS_1=\{\bs_i, 2\leq i\leq14\}$ with
\begin{eqnarray*}
\bs_i=\left(\phi(1+\alpha_{i}),\phi(\theta+\alpha_{i}),\phi(\theta^2+\alpha_{i}),\cdots, \phi(\theta^{79}+\alpha_{i})\right).
\end{eqnarray*}
It is easy to check that FHS set $\mS_1$ is an optimal $(80,13,6;14)$ FHS set with new parameters.
}\end{ex}

\begin{ex}\label{e2}{\rm
Let $q=3$, $r=2$, $m=6$, and $t=2$.

Let $\theta$ be a generator of $\F_{3^6}^*$. Let $G=\F_{3}^*$ be the multiplicative subgroup of $\F_3$ with $|G|=2$. Let $V$ be an $\F_{3}$-subspace of $\F_{3^6}$ of dimension $t=2$.

Choose $\alpha=$$\{\alpha_i\in \F_{3^6}:2\leq i\leq41\}$ satisfying that $V$, $\{\alpha_ig+V:{2\leq i\leq 41, g\in G}\}$ are $81$ pairwise distinct cosets of $V$ as defined in Lemma \ref{l1}.

Choose $\phi(x)$ as defined in Lemma \ref{l2}. For every $ \Ga_i\in\{\Ga_2,\cdots, \Ga_{41}\}$, we can obtain an FHS set $\mS_2=\{\bs_i, 2\leq i\leq41\}$ with
\begin{eqnarray*}
\bs_i=\left(\phi(1+\alpha_{i}),\phi(\theta+\alpha_{i}),\phi(\theta^2+\alpha_{i}),\cdots, \phi(\theta^{727}+\alpha_{i})\right).
\end{eqnarray*}
It is easy to check that FHS set $\mS_2$ is an optimal $(728,40,18;41)$ FHS set with new parameters.
}\end{ex}

\begin{ex}\label{e3}{\rm
Let $q=7$, $r=3$, $m=3$, and $t=1$.

Let $\theta$ be a generator of $\F_{7^3}^*$. Let $G=\F_{7}^*$ be the multiplicative subgroup of $\F_7$ with $|G|=3$. Let $V$ be an $\F_{7}$-subspace of $\F_{7^3}$ of dimension $t=1$.

Choose $\alpha=$$\{\alpha_i\in \F_{7^3}:2\leq i\leq17\}$ satisfying that $V$, $\{\alpha_ig+V:{2\leq i\leq 17, g\in G}\}$ are $49$ pairwise distinct cosets of $V$ as defined in Lemma \ref{l1}.

Choose $\phi(x)$ as defined in Lemma \ref{l2}. For every $ \Ga_i\in\{\Ga_2,\cdots, \Ga_{17}\}$, we can obtain an FHS set $\mS_3=\{\bs_i : 2\leq i\leq17\}$ with
\begin{eqnarray*}
\bs_i=\left(\phi(1+\alpha_{i}),\phi(\theta+\alpha_{i}),\phi(\theta^2+\alpha_{i}),\cdots, \phi(\theta^{341}+\alpha_{i})\right).
\end{eqnarray*}
It is easy to check that FHS set $\mS_3$ is an optimal $(342,16,21;17)$ FHS set with new parameters.
}\end{ex}

%%%%%%%%%%%%%%%extension construction%%%%%%%%%%%%%%%
\section{Recursive Constructions of FHS sets with new parameters}

Recursive constructions have been proposed in~\cite{Chung14,Ren14,Xu16,Bao16,Niu18} to generate new families of optimal FHS sets under certain conditions. The framework of recursive constructions in~\cite{Niu18} can be used to get FHS set with new parameters, which increase the length and alphabet size of the original FHS set, but preserve its family size and maximum Hamming correlation. In this section, by applying the FHS set $\mS$ in Theorem \ref{FHS1} to a recursive construction, we  obtain some new classes of FHS sets.

Let us recall a recursive construction first. Denote by $\m(\mS)$ the  maximum appearance number of frequency slots of an FHS set $\mS$.
\begin{lemma}\label{lrec}
Whenever there exist an $(N,M,H_m;l)$ FHS set $\mS$ and an $(n,s;v)$ OC sequence set $\mC$ with $s\geq\m(\mS)$, there is an $(nN,M,H_m;vl)$ FHS set $\mX$.
\end{lemma}
To apply Lemma \ref{lrec}, we have to give an upper bound on $\m(\mS)$.
\begin{lemma}\label{lmax} For $r\geq2$,
the maximum appearance number $\m(\mS)$ of frequency slots of the FHS set $\mS$ in Theorem \ref{FHS1} is upper bounded by $q^m-q^t-1$.
\end{lemma}
\begin{proof}
For $1\le i\le \ell$, denote by $A_i$ the  set $\{\theta^k+\Ga_i:\; 0\le k\le q^m-2\}$.
By Lemma \ref{l3}, to count appearance of an element, we have to count the appearance number $n_{ij}$ of elements of $\bigcup_{g\in G}(\Ga_ig+V)$ in $A_j$ for all $1\le i\le \ell$ and $2\le  j\le \ell$. Then we have $\m(\mS)=\max\limits_{1\le i\le \ell}\{\sum_{j=2}^\ell n_{ij}\}$.

 For $i=1$ and $2\le  j\le \ell$, every element in $V$ appears in $A_j$ once. Thus, the total appearance number of elements of $\bigcup_{g\in G}(\Ga_1g+V)$ in $\bigcup_{j=2}^\ell A_j$ is $|V|(\ell-1)=q^t\frac{q^{m-t}-1}r< q^m-q^t-1$ for $r\geq2$.

For $2\le i\le \ell$,   every element in $\bigcup_{g\in G}(\Ga_ig+V)$ appears in $A_j$ once for $j\neq i$. For $j=i$,  every element in $\bigcup_{g\in G}(\Ga_ig+V)$ except for $\Ga_i$ appears in $A_i$ once and $\Ga_i$ does not appear in $A_i$. Thus, the total appearance number of elements of $\bigcup_{g\in G}(\Ga_ig+V)$ in $\bigcup_{j=2}^\ell A_j$ is $r|V|(\ell-1)-1=rq^t\left(\frac{q^{m-t}-1}r\right)-1= q^m-q^t-1$. The proof is completed.
\end{proof}

The recursive constructions extend the FHS set $\mS$ in the above section by choosing different one-coincide sequence sets. There are some known constructions of OC sequence sets~\cite{OC,Cao06,W15,Ren17,Niu18}. Based on the framework of the recursive construction in~\cite{Niu18}, we can obtain FHS sets with new parameters by combing the  FHS set in Theorem \ref{FHS1} with OC sequence sets.

Let $e=\frac{q^{m-t}-1}{r}$, then we have the following theorem and corollaries.

\begin{theorem}\label{FHS-ex} Put $e=\frac{q^{m-t}-1}r$. Then whenever there is an  $(n,s;v)$ OC sequence set with $s>q^m-q^t-1$, there exists a
 $\left(n(q^m-1), e, rq^t; v(e+1)\right)$ FHS set $\mX$. Furthermore, $\mX$ is optimal if
  \[\left\lceil\frac{n(q^m-1)e-v(e+1)}{n(q^m-1)e-1}\frac{n(q^m-1)}{v(e+1)}\right\rceil=
\left\lceil\frac{(q^m-1)e-(e+1)}{(q^m-1)e-1}\frac{(q^m-1)}{e+1}\right\rceil.\]
\end{theorem}
\begin{proof} The desired result follows from Theorem \ref{thm:1.1} and Lemma \ref{lrec}.
\end{proof}

{\bf Proof of Theorem \ref{thm:1.2}:}
\begin{proof}
By applying the $(k,\lpf(k)-1; k)$ OC sequence set given in~\cite{Chung14,Bao16,Niu18}, the $(p-1,p;p)$ OC sequence set given in~\cite{Chung14,Bao16,Niu18} and the $\left(k(p-1),\min\{\lpf(k)-1,p\}; kp\right)$ OC sequence set given in~\cite{Chung14,Niu18}, respectively to Theorem \ref{FHS-ex}, we obtain the desired FHS sets in Table 1.
\end{proof}

We illustrate Theorem \ref{thm:1.2}  by the following examples.
\begin{ex}{\rm
Choose the optimal $(80,13,6;14)$ FHS set $\mS_1$ in Example \ref{e1}. We can obtain the maximum appearance number of frequency slots in $\mS_1$ is $77$. Thus, the desired recursive result of FHS sets with new parameters are as follows.
\\$1)$ There is an optimal $(79\times80,13,6;79\times14)$ FHS set by applying the $(79,78;79)$ OC sequence set.
\\$2)$ There is an optimal $(80\times80,13,6;81\times14)$ FHS set by applying the $(80,81;81)$ OC sequence set.
\\$3)$ There is an optimal $(79\times80\times80,13,6;79\times81\times14)$ FHS set by applying the $(79\times80,78;79\times81)$ OC sequence set.
}\end{ex}

\begin{ex}{\rm
Choose the optimal $(728,40,18;41)$ FHS set $\mS_2$ in Example \ref{e2}. We can obtain the maximum appearance number of frequency slots in $\mS_2$ is $719$. Thus, the desired recursive result of FHS sets with new parameters are as follows.
\\$1)$ There is an optimal $(727\times728,40,18;727\times41)$ FHS set by applying the $(727,726;727)$ OC sequence set.
\\$2)$ There is an optimal $(728\times728,40,18;729\times41)$ FHS set by applying the $(728,729;729)$ OC sequence set.
\\$3)$ There is an optimal $(727\times728\times728,40,18;727\times729\times41)$ FHS set by applying the $(727\times728,726;727\times729)$ OC sequence set.
}\end{ex}

\section{Conclusion}
In this paper, we present new construction of optimal FHS sets by mixing both multiplicative and additive groups structures of finite fields simultaneously. The construction provides a new family of optimal $\left(q^m-1,\frac{q^{m-t}-1}{r},rq^t;\frac{q^{m-t}-1}{r}+1\right)$ frequency hopping sequence sets archiving the Peng-Fan bound. Thus, the FHS sets constructed in literatures using either multiplicative groups or additive groups of finite fields are all included in our family. It should be noted that our construction not only includes some constructions in literatures as special cases, but also gives new and flexible parameters due to the free choice of $t$. In addition, some other FHS sets can be obtained via the well-known recursive constructions through one-coincidence sequence set. As a result, our constructions allow a great flexibility of choosing FHS sets for a given frequency-hopping spread spectrum system.

\end{document}